\documentclass{llncs}
\usepackage{latexsym}
\usepackage{graphicx}
\usepackage{algorithm}
\usepackage[noend]{algcompatible}

\newcommand{\eps}{\epsilon} 
\title{Optimal Online Two-way Trading with Bounded Number of Transactions}

\author{Stanley P. Y. Fung
}

\institute{Department of Informatics, University of Leicester,
Leicester LE1 7RH, United~Kingdom.
\email{pyf1@leicester.ac.uk}
}

\begin{document}
\maketitle

\begin{abstract}
We consider a two-way trading problem, where investors buy and sell a stock
whose price moves within a certain range. Naturally they want to
maximize their profit. Investors can perform up to $k$ trades,
where each trade must involve the full amount.
We give optimal algorithms for three different models which differ in the
knowledge of how the price fluctuates. In the first model, there are global
minimum and maximum bounds $m$ and $M$.
We first show an optimal lower bound of $\varphi$ (where $\varphi=M/m$)
on the competitive ratio for one trade, which is the bound achieved by
trivial algorithms. Perhaps surprisingly, when we consider more 
than one trade, we can give a
better algorithm that loses only a factor of $\varphi^{2/3}$ (rather than 
$\varphi$) per additional trade. Specifically, for $k$ trades the
algorithm has competitive ratio $\varphi^{(2k+1)/3}$. 
Furthermore we show that this ratio is the best possible by giving a matching
lower bound.
In the second model, $m$ and $M$ are not known in advance, and
just $\varphi$ is known. We show that this only costs us an extra factor of
$\varphi^{1/3}$, i.e., both upper and lower bounds become $\varphi^{(2k+2)/3}$.
Finally, we consider the bounded daily return model where instead of a 
global limit, the fluctuation from one day to the next is bounded, and
again we give optimal algorithms, and interestingly one of them resembles common
trading strategies that involve stop loss limits.
\end{abstract}

\section{Introduction}

\paragraph{The model.}
We consider a scenario commonly faced by investors. The price of a stock
varies over time. In this paper we use a `day' as the smallest unit of
time, so there is one new price each day.
Let $p(i)$ be the price at day $i$. The investor has some initial
amount of money. Over a time horizon of finite duration $T$, the investor wants
to make a bounded number of trades of this one stock. 
Each trade $(b,s)$ consists of a 
buy transaction at day $b$, followed by a sell transaction at day $s$
where $s>b$. (Thus one trade consists of two transactions.)
Both transactions are `all-in': when buying, the investor uses all the money 
available, and when selling all stock they currently own is sold.
A sale must be made before the next purchase can take place.
Also, no short selling is allowed, i.e., there can be no selling if
the investor is not currently holding stock.
When the end of the time horizon is reached, i.e., on the last day,
no buying is allowed and the investor must sell off all the stocks that
they still hold back to cash at the price of the day.

There are a number of rationales for considering a bounded number of trades
and/or that trades must involve all the money available.
Individual, amateur investors typically do not want to make 
frequent transactions due to high transaction fees. 
Often transaction fees have a fixed component (i.e., a fixed amount or a
minimum tariff per transaction, irrespective of the trading amount)
which makes transaction fees disproportionally high for small trades.
Frequent trading also requires constant monitoring
of the markets which amateur investors may not have the time or resources for;
often they only want to change their investment portfolios every now and then.
Also, for investors with little money available, 
it is not feasible or sensible to divide them into smaller pots of money,
in arbitrary fractions as required by some algorithms.
The finiteness of the time horizon (and that its length is possibly 
unknown as well) corresponds
to situations where an investor may be forced to sell and leave the market
due to unexpected need for money elsewhere, for example.

Each trade with a buying price of $p(b)$ and a selling price of $p(s)$ gives a 
{\it gain} of $p(s)/p(b)$. This represents how much the investor has after the
trade if they invested 1 dollar in the beginning.
Note that this is a ratio, and can be less than 1, meaning there is a loss,
but we will still refer to it as a `gain'. 
If a series of trades are made, the overall gain or the {\it return} of
the algorithm is the product of the gains of each of the individual
trades. This correctly reflects the fact that all the money after each trade
is re-invested in the next.

Since investors make decisions without knowing future stock prices,
the problem is {\it online} in nature.
We measure the performance of online algorithms with {\it competitive 
analysis}, 
i.e., by comparing it with the optimal offline algorithm OPT that
knows the price sequence in advance and can therefore make optimal decisions.
The {\it competitive ratio} of an online algorithm ONL is the worst
possible ratio of the return of OPT to the return of ONL, 
over all possible input (price) sequences.
The multiplicative nature of the definition of the return 
(instead of specifying a negative value for a loss) means that
the competitive ratio can be computed in the normal way in the case of a
loss: for example, if OPT makes a gain of 2 and ONL makes a `gain' of 1/3, 
then the competitive ratio is 6.

\paragraph{Three models on the knowledge of the online algorithm.}
We consider three different models on how the price changes, or equivalently,
what knowledge the online algorithm has in advance. In the first model,
the stock prices are always within a range [$m..M$], 
i.e., $m$ is the minimum possible price and $M$ the maximum possible price. 
Both $m$ and $M$ are known to the online algorithm up front.
In the second model, the prices still fluctuate within this range, but
$m$ and $M$ are not (initially) known; instead only their ratio
$\varphi = M/m$, called the {\it fluctuation ratio}, is known. 
In both these models the length of the time horizon (number of days)
is unknown (until the final day arrives).
In the third model, called the {\it bounded daily return model},
there is no global minimum or maximum price. Instead,
the maximum fluctuation from day to day is bounded: namely, the price
$p(i+1)$ of the next day is bounded by the price $p(i)$ of the current
day by $p(i)/\beta \le p(i+1) \le \alpha p(i)$ for some $\alpha, \beta>1$.
This means the prices cannot suddenly change a lot. 
Many stock markets implement the so-called `circuit breakers'
where trading is stopped when such limits are reached.
Here $\alpha, \beta$ and the trade duration $T$ 
are known to the online algorithm.
All three models are well-established in the algorithms
literature; see e.g. \cite{CKLW01,EFKT01}.

\paragraph{Previous results and related work.}

Financial trading and related problems are obviously important topics
and have been much studied from the online algorithms perspective.
A comprehensive survey is given in \cite{SORMS}.
Here we only sample some of the more important results and those closer to
the problems we study here.
In the {\it one-way search problem}, 
the online player chooses one moment of time to make a single transaction
from one currency to another currency. Its return is simply the price at
which the transaction takes place.
A reservation price (RP) based policy is to buy as soon as the price reaches
or goes above a pre-set {\it reservation price}.
It is well-known that, if $m$ and $M$ are known, the RP policy with 
a reservation price of $\sqrt{Mm}$ is optimal and achieves
a competitive ratio of $\sqrt{\varphi}$.
If only $\varphi$ is known, then no deterministic algorithm can achieve a ratio 
better than $\varphi$. With the help of randomization, however,
a random mix of different RPs gives a competitive ratio of $O(\log \varphi)$
if $\varphi$ is known. Even if $\varphi$ is not known, a competitive
ratio of $O(\log \varphi \cdot \log^{1+\eps}(\log \varphi))$ can be achieved. See 
\cite{EFKT01} for all the above results and more discussions.

In the {\it one-way trading problem}, the objective is again to maximize
the final amount in the other currency, but there can be multiple transactions,
i.e., not all the money has to be traded in one go.
(This distinction of terminology between search and trading is used in
\cite{EFKT01}, but is called non-preemptive vs. preemptive in \cite{SORMS}.
We prefer calling them {\it unsplittable} vs. {\it splittable} here.)
The relation between one-way trading and randomized algorithms for
one-way search is described in \cite{EFKT01}.
Many variations of one-way search or one-way trading problems have since been 
studied; 
some examples include the bounded daily return model \cite{CKLW01,ZXZD12},
searching for $k$ minima/maxima instead of one \cite{LPS09}, 
time-varying bounds \cite{DHT09}, 
unbounded prices \cite{CFG+15}, 
search with advice complexity \cite{IWOCA16}, etc.

What we study here, however, is a {\it two-way} version of the 
unsplittable trading problem%
\footnote{In the terminology of \cite{EFKT01} this should be called 
`{\it two-way search}', but we feel that the term does not convey its
application in stock market trading.}%
, which is 
far less studied. Here the online player has to first convert from one currency
(say cash) to another (a stock), hopefully at a low price, and then convert
back from the stock to cash at some later point, hopefully at a high price.
All the investment must be converted back to the first currency when or before
the game ends. This model is relevant where investors are only interested in 
short term, speculative gains.
For the models with known $m,M$ or known $\varphi$ and with one trade,
Schmidt et al. \cite{ENDM} gave a $\varphi$-competitive algorithm;
it uses the same RP for buying and selling. But consider the {\sc Do-Nothing}
algorithm that makes no trades at all. Clearly it is also
$\varphi$-competitive as ONL's gain is 1 and
OPT's gain is at most $\varphi$ (if the price goes from $m$ to $M$).
A number of common trading strategies, such as those based on moving averages,
were also studied in \cite{SORMS}. It was shown 
that they are $\varphi^2$-competitive (and not better), which are
therefore even worse.
It is easy to show that these algorithms have competitive ratios
$\varphi^k$ and $\varphi^{2k}$ respectively when extended to $k$ trades.
Schroeder et al. \cite{CoDIT16} gave some algorithms for the bounded daily
return model, without limits on the number of trades.
However, most of these algorithms 
tend to make decisions that are clearly bad, have the
worst possible performance (like losing by the largest
possible factor every day throughout), 
and have competitive ratios no better than what is given by
{\sc Do-Nothing}.

\paragraph{Our results.}

In this paper we consider the two-way unsplittable trading problem 
where a bounded number $k$ of trades are permitted, and derive optimal algorithms.
First we consider the model with known $m$ and $M$.
We begin by considering the case of $k=1$. Although some naive algorithms 
are known to be $\varphi$-competitive and seemingly nothing better is 
possible, we are not aware of any matching general lower bound.
We give a general lower bound of $\varphi$, 
showing that the naive algorithms cannot be improved.
The result is also needed in subsequent lower bound proofs.

It may be tempting to believe that nothing can beat the naive algorithm also
for more trades. Interestingly, we prove that for $k \ge 2$ this is not true.
While naive algorithms like {\sc Do-Nothing} are no better than
$\varphi^k$-competitive, we show that a
reservation price-based algorithm is $\varphi^{(2k+1)/3}$-competitive.
For example, when $k=2$, it is $\varphi^{5/3}$-competitive instead of
trivially $\varphi^2$-competitive.
Furthermore, we prove a matching lower bound, showing that the algorithm is
optimal.

Next, we consider the model where only $\varphi$ is known, and give
an algorithm with a competitive ratio of $\varphi^{(2k+2)/3}$, i.e.,
only a factor $\varphi^{1/3}$ worse than that of the preceding model.
Again we show that this bound is optimal.

Finally we consider the bounded daily return model, and give two optimal
algorithms where the competitive ratio depends on $\alpha, \beta$ and $T$.
For example, with one trade and in the symmetric case $\alpha=\beta$,
the competitive ratio is $\alpha^{2T/3}$. While this is exponential in
$T$ (which is unavoidable), naive algorithms could lose up to a factor of 
$\max(\alpha,\beta)$ every day, and
{\sc Do-Nothing} has a competitive ratio of $\alpha^T$.
One of the algorithms uses the `stop loss / lock profit' strategy 
commonly used in real trading; as far as we are aware, this is the first time 
where competitive analysis justifies this common stock market trading strategy,
and in fact suggests what the stop loss limit should be.

\section{Known $m$ and $M$}

In this section, where we consider the model with known $m$
and $M$, we can without loss of generality assume that $m=1$. This is what
we will do to simplify notations. It also means $M$ and $\varphi$ are 
equal and are sometimes used interchangeably.

\begin{theorem} \label{thm:1tradeLB}
For $k=1$, no deterministic algorithm has a competitive ratio better than 
$\varphi^{1-\eps}$, for any $\eps>0$.
\end{theorem}
\begin{proof}
Choose $n = \lceil 1/\eps \rceil$ and define $v_i = M^{i/n}$ for 
$i=0, 1, \ldots, n$.
The following price sequence is released until ONL buys:
$v_{n-1}, M, v_{n-2}, M, \ldots, v_i$, $M, \ldots, v_1$, $M, v_0$.
If ONL does not buy at any point, or buys at price $M$, then its return is 
at most 1. Then OPT buys at $v_1$ and sells at $M$ to get a return of $M^{1-1/n}$.
So suppose ONL buys at $v_i$ for some $1 \le i \le n-1$. 
(It cannot buy at $v_0$ as it is the last time step.)
As soon as ONL bought, the rest of the sequence is not released; instead
the price drops to $m$ and the game ends.
ONL's return is $m/v_i$. If $i=n-1$, then OPT makes no trade and its return is 1,
so competitive ratio = $v_{n-1}/m = M^{1-1/n}$. Otherwise, if $i<n-1$,
OPT buys at $v_{i+1}$ (two days before ONL's purchase)
and sells at the next day at price $M$, giving a return of $M/v_{i+1}$.
The competitive ratio is therefore $Mv_i/(mv_{i+1}) = M^{1-1/n}$.

Thus in all cases the competitive ratio is at least 
$M^{1-1/n} \ge \varphi^{1-\eps}$.
\qed \end{proof}

Note that the proof does not require ONL to use only one trade: it cannot
benefit even if it is allowed to use more trades. This fact will
be used later in Theorems~\ref{thm:ktradeLB} and~\ref{thm:ktradeLB2}.

For $k>1$, we analyze the following algorithm:

\begin{algorithm}[h]
\caption{The reservation price algorithm, with known price range $[m..M]$}
\label{alg:1}
\begin{algorithmic}
\STATE Upon release of the $i$-th price $p(i)$:
\IF{$i=T$} 
  \IF{currently holding stock}
    \STATE sell at price $p(T)$ and the game ends
  \ENDIF
\ELSE
  \IF{$p(i) \le M^{1/3}$ and currently not holding stock and not used up all trades} 
    \STATE buy at price $p(i)$
  \ELSIF{$p(i) \ge M^{2/3}$ and currently holding stock}
    \STATE sell at price $p(i)$
  \ENDIF
\ENDIF
\end{algorithmic}
\end{algorithm}

\begin{theorem} \label{thm:ktradeUB}
Algorithm \ref{alg:1} has competitive ratio $\varphi^{(2k+1)/3}$, for $k \ge 1$.
\end{theorem}
\begin{proof} 
First we make a few observations. 
We call a trade {\it winning} if its gain is higher than 1, and {\it losing}
otherwise.
Any winning trade made by ONL has a gain of at least $M^{1/3}$.
If the algorithm makes a losing trade, it must be a forced sale at the end
and the gain is not worse than $M^{-1/3}$. Moreover, it follows that
the algorithm cannot trade anymore after it.

We consider a number of cases separately based on the sequence of win/loss
trades. If we use $W$ and $L$ to denote a winning and a losing trade
respectively, then following the above discussion, the possible cases are: 
nil (no trade), $L$, $W^{j}L$ for $1 \le j \le k-1$, and
$W^{j}$ for $1 \le j \le k$. (Here $W^j$ denotes a sequence of $j$ consecutive
$W$'s.)

\begin{figure}
\begin{center}
\includegraphics[width=10.5cm]{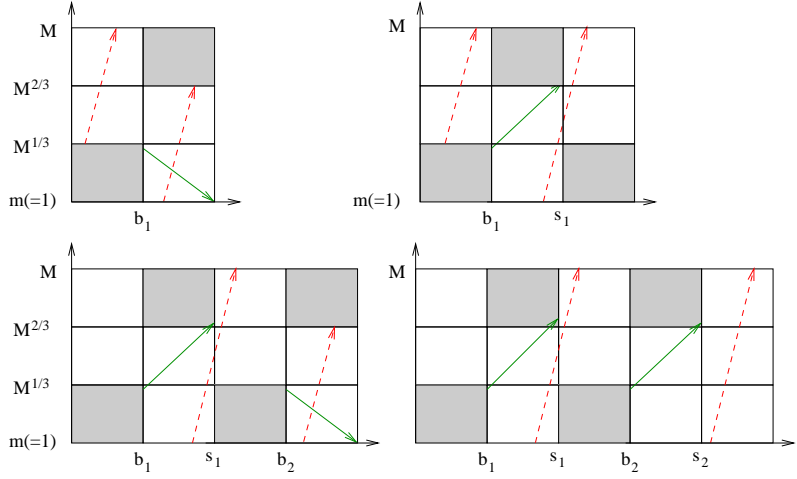}
\caption{
Four cases illustrated, for $k=2$. Horizontal
axis is time, vertical axis is price. Shaded regions are the regions where
the prices cannot fall into. 
Green solid arrows depict possible buying and selling actions of ONL, 
red dashed arrows for OPT. Top left: case L, Top right: case W,
Bottom left: case WL, Bottom right: case WW.}
\label{fig:1}
\end{center}
\end{figure}

\begin{description}

\item[Case nil:] 
Since ONL has never bought, the prices were never at or below $M^{1/3}$ 
(except possibly
the last one, but neither OPT nor ONL can buy there) and hence OPT's return cannot
be better than $(M/M^{1/3})^k = M^{2k/3}$. ONL's return is 1. 
So the competitive ratio is at most $M^{2k/3}$.

\item[Case $L$:] Suppose ONL buys at time $b_1$ and is forced to sell at the end.
The prices before $b_1$ cannot be lower than $M^{1/3}$ (or else it would have 
bought) and the prices after $b_1$ cannot be higher than $M^{2/3}$
(or else it would have sold). Thus, it is easy to see (Figure~\ref{fig:1})
that OPT cannot make any trade with gain higher than $M^{2/3}$.
So the competitive ratio is at most $(M^{2/3})^k / M^{-1/3} = M^{(2k+1)/3}$.

\item[Case $W$:] Suppose ONL buys at time $b_1$ and sells at time $s_1$. Then
before $b_1$, the prices cannot be below $M^{1/3}$; between $b_1$ and $s_1$, 
the prices cannot be higher than $M^{2/3}$; and after $s_1$, the prices 
cannot be lower than $M^{1/3}$. It can be seen from Figure~\ref{fig:1} that 
OPT can make at most one trade with gain $M$ (crossing time $s_1$); any other 
trade it makes must be of gain at most $M^{2/3}$. So the competitive ratio 
is at most $M (M^{2/3})^{k-1}/M^{1/3} = M^{2k/3}$.

\item[Case $W^{j}L$, $1 \le j \le k-1$:]
Similarly, we can partition the timeline into regions (Figure~\ref{fig:1}), 
from which we can see that OPT can make at most $j$ trades of gain $M$ and 
the rest have gain at most $M^{2/3}$. Thus competitive ratio = 
$(M^j (M^{2/3})^{k-j})/((M^{1/3})^j M^{-1/3}) = M^{(2k+1)/3}$.

\item[Case $W^j$, $1 < j \le k-1$:] 
This can only be better than the previous case,
as OPT again can make at most $j$ trades of gain $M$ and the rest have
gain at most $M^{2/3}$, but ONL's return is better than the previous case.

\item[Case $W^k$:] 
In this case the competitive ratio is simply $M^k / (M^{1/3})^k = M^{2k/3}$. \qed
\end{description}
\end{proof}

\begin{theorem} \label{thm:ktradeLB}
No deterministic algorithm has a competitive ratio better than 
$\varphi^{(2k+1)/3-\eps}$, for any $\eps>0$ and $k \ge 1$.
\end{theorem}
\begin{proof}
The prices are released in up to $k$ rounds. The final round $k$ is a special
round. For all other rounds, we maintain the following invariants. 
For each $1 \le i \le k-1$, just before the $i$-th round starts, 
OPT completed exactly $i-1$ trades, is holding no stock, and
accumulated a return of exactly $M^{i-1}$, while ONL completed at most $i-1$
trades, is holding no stock, and
accumulated a return of at most $M^{(i-1)/3}$.
So the competitive ratio up to this point is at least $M^{2(i-1)/3}$.

For any $i<k$, round $i$ begins with the price sequence 
$M^{1/3}, M, M^{1/3}, M, \ldots$ until either ONL buys or $k-i$ such pairs 
of oscillating prices have been released. 
If ONL does not buy at any point, then the round ends.
Clearly, ONL maintains its variants.
OPT makes $k-i$ trades giving a total gain of $(M^{2/3})^{k-i}$ in this round, 
and thus the accumulated competitive ratio is $M^{2(k-1)/3}$. 
It also used $(i-1)+(k-i) = k-1$ trades.
Any remaining intermediate rounds are then skipped and we jump directly to the
special last round $k$.

Otherwise, assume ONL buys at one of the $M^{1/3}$ prices
($M$ is clearly even worse). 
The rest of that sequence will not be released. Instead, the price
sequence that follows is $m, M^{2/3}, m, M^{2/3}, \ldots$ until either ONL
sells or $k-i+1$ such pairs of oscillating prices were released.
If ONL does not sell at any of these, then the price drops to $m$ and the game 
ends (with no further rounds, not even the special round). ONL's gain in 
this round is $M^{-1/3}$. OPT uses all its remaining $k-i+1$ trades and gains 
$(M^{2/3})^{k-i+1}$.
Combining with the previous rounds, the competitive ratio is at most
$M^{2(i-1)/3} M^{2(k-i+1)/3} / M^{-1/3} = M^{(2k+1)/3}$.

Otherwise ONL sells at one of the $M^{2/3}$ prices ($m$ is even worse). 
The rest of that sequence will not be released; instead the price goes up to 
$M$ and this round ends. OPT's gain in this round is $M$ by making one trade 
from $m$ to $M$; ONL gains $M^{1/3}$. 
Thus the invariants are maintained and we move on to the next round.
(Regarding the invariant that ONL is not holding stock at the end of the
round, we can assume w.l.o.g. that ONL does not buy at the last price $M$, 
since clearly it cannot make a profit doing so. In any case, even if it does 
buy, it can be treated as if it were buying at the beginning of the next round.)

Finally, if we arrive at round $k$, then the same price sequence as in
Theorem~\ref{thm:1tradeLB} is used to give 
an additional factor of $M^{1-\eps}$ to the competitive ratio. 
Note that at the start of this round, OPT has one trade left, and ONL has one
or more trades left, but that will not help. Thus the 
competitive ratio is not better than $M^{2(k-1)/3+1-\eps} = M^{(2k+1)/3-\eps}$.
\qed \end{proof}

\section{Known $\varphi$ only}

For $k=1$ {\sc Do-Nothing} is clearly still $\varphi$-competitive, and
Theorem~\ref{thm:1tradeLB} still applies here, so we focus on $k>1$.
We adapt Algorithm~\ref{alg:1} by buying only when it is certainly 
`safe', i.e., when it is certain that the price is within the lowest 
$\varphi^{1/3}$ of the actual price range, and sells when it gains 
$\varphi^{1/3}$. The formal description is given 
in Algorithm~\ref{alg:2}.
Let $M_t$ be the maximum price observed up to and including day $t$.
Note that $M_t$ is a stepwise increasing function of $t$.

\begin{algorithm}[h]
\caption{The reservation price algorithm, with known $\varphi$}
\label{alg:2}
\begin{algorithmic}
\STATE Upon release of the $i$-th price $p(i)$:
\IF{$i=T$}
  \IF{currently holding stock}
    \STATE sell at price $p(T)$ and the game ends.
  \ENDIF
\ELSE
  \IF{$i=1$}
    \STATE $M_1 := p(1)$
  \ELSE
    \STATE $M_i := \max(M_{i-1}, p(i))$
  \ENDIF
  \IF{$p(i) \le M_i/\varphi^{2/3}$ and currently not holding stock and
      not used up all trades} 
    \STATE buy at price $p(i)$
  \ELSIF{currently holding stock bought at price $p(b)$ and $p(i) \ge 
\varphi^{1/3}p(b)$}
    \STATE sell at price $p(i)$
  \ENDIF
\ENDIF
\end{algorithmic}
\end{algorithm}

\begin{theorem} \label{thm:ktradeUB2}
Algorithm \ref{alg:2} has competitive ratio $\varphi^{(2k+2)/3}$, for any $k \ge 2$.
\end{theorem}
\begin{proof}
Clearly ONL gets the same as in Theorem~\ref{thm:ktradeUB}:
each winning trade has gain at least $\varphi^{1/3}$ and a losing trade, 
which can only appear as the last trade, has gain at least $\varphi^{-1/3}$. 
The difference is in how we bound OPT's gain.

In the case of $W^k$ (ONL makes $k$ winning trades) then the same argument
as Theorem~\ref{thm:ktradeUB} applies, so in the following we only consider the
case where ONL did not use up all its trade, i.e., it is always able to buy if
it is not holding.

A {\it sell event} happens at a day when ONL sells and makes a 
profit (i.e., excludes the forced sale at the end).
An {\it M-change event} happens at day $t$ if $M_t \neq M_{t-1}$.
Each OPT trade ($b^*, s^*)$ can be classified into one of the following types:

\begin{enumerate}

\item[(1)] There is at least one sell event during $[b^*, s^*]$.
Clearly the number of such OPT trades is limited by the number of sell events.
Each such trade can gain up to $\varphi$.

\item[(2)] There is no sell event during $[b^*, s^*]$, and at $b^*$ ONL is holding
or buying.
Suppose ONL's most recent purchase is at time $b \le b^*$. Then
$p(b) \le M_b/\varphi^{2/3} \le M_{b^*}/\varphi^{2/3}$. 
It is holding stock throughout
and still did not sell at $s^*$ (or is forced to sell if $s^*$ is the last day), 
hence $p(s^*) < p(b) \varphi^{1/3} \le M_{b^*}/\varphi^{1/3}$.
But clearly $p(b^*) \ge M_{b^*}/\varphi$, hence the gain of OPT is at most
$\varphi^{2/3}$.

\item[(3)] There is no sell event during $[b^*, s^*]$, 
at $b^*$ ONL is neither holding nor buying, 
and there is no M-change event in $(b^*, s^*]$.
We have $p(b^*) > M_{b^*}/\varphi^{2/3}$ as otherwise ONL would have bought at 
$b^*$. Clearly $p(s^*) \le M_{s^*} = M_{b^*}$.
Hence the gain of OPT is $p(s^*)/p(b^*) < \varphi^{2/3}$.

\item[(4)] There is no sell event during $[b^*, s^*]$, 
at $b^*$ ONL is neither holding nor buying, 
and there is/are M-change event(s) in $(b^*, s^*]$.
Suppose there are a total of $x$ such OPT trades, 
$(b^*_1, s^*_1), (b^*_2, s^*_2), \ldots, (b^*_x, s^*_x)$, in chronological order.
Note that $p(b^*_i) > M_{b^*_i}/\varphi^{2/3}$ or else ONL would have bought 
at $b^*_i$. So for all $i$,
$p(b^*_{i+1}) > M_{b^*_{i+1}}/\varphi^{2/3} \ge M_{s^*_i}/\varphi^{2/3} \ge 
p(s^*_i)/\varphi^{2/3}$.
Thus the total gain of these $x$ trades is
\[ \prod_{i=1}^x \frac{p(s^*_i)}{p(b^*_i)}
= \frac{1}{p(b^*_1)} \frac{p(s^*_1)}{p(b^*_2)} \cdots
\frac{p(s^*_{x-1})}{p(b^*_x)} \frac{p(s^*_x)}{1}
< \frac{p(s^*_x)}{p(b^*_1)}(\varphi^{2/3})^{x-1}
\le \varphi (\varphi^{2/3})^{x-1} = 
\varphi^{(2x+1)/3}. \]
\end{enumerate}

Suppose ONL makes $y$ winning trades and one losing trade. 
Then OPT makes at most $y$ trades of type (1), gaining at most $\varphi^{y}$ 
from those. Then, if $x$ of OPT's trades are of type (4), they in total gives
another gain of at most $\varphi^{(2x+1)/3}$. The remaining trades are
of types (2) and (3), gaining $\varphi^{2/3}$ each. 
The competitive ratio is therefore at most
\[ \frac{ \varphi^y \varphi^{(2x+1)/3} \varphi^{2(k-x-y)/3} }
        { \varphi^{y/3} \varphi^{-1/3} }
= \varphi^{ (2k+2)/3 }. \]

If ONL makes $y<k$ winning trades and no losing trade, the competitive ratio
can only be better, as OPT's return is as above but ONL's is $\varphi^{1/3}$ better.
\qed \end{proof}

\begin{theorem} \label{thm:ktradeLB2}
No deterministic algorithm is better than $\varphi^{(2k+2)/3-\eps}$-competitive,
for any $\eps>0$ and $k \ge 2$.
\end{theorem}
\begin{proof}
Again there will be a number of rounds.
Round 1 is special, in that OPT will get a factor of $\varphi$ better than ONL
but will afterwards reveal knowledge of $m$ and $M$.
Rounds 2 to $k$ are then similar to Theorem~\ref{thm:ktradeLB}.

Round 1:
The first price is $1$. If ONL does not buy, then the price goes up to
$\varphi$. OPT makes one trade and gains $\varphi$. Now we know the range
is [$1..\varphi$], and we can assume w.l.o.g. that ONL does not buy at $\varphi$. 
Then the round ends. At the end of this round, both OPT and ONL are not holding
stock, OPT made one trade and ONL none, but ONL
is a factor of $\varphi$ behind in the return.

Otherwise, if ONL buys at $1$, then the subsequent price sequence is 
$1/\varphi, 1$, $1/\varphi, 1, \ldots$ for up to $k$ such pairs, until ONL sells. 
Now we know the range is [$1/\varphi..1$]. Without loss of generality 
we can assume ONL does not sell at $1/\varphi$ since it is clearly the lowest
possible price. If ONL does not sell at any point, then the game ends
with no further rounds. OPT makes $k$ trades gaining $\varphi^k$, and ONL's gain
is 1. The competitive ratio is $\varphi^k$, which is at least 
$\varphi^{(2k+2)/3}$. 
If ONL sells at some point with price $1$, then the sequence stops and 
this round ends. OPT buys at $1/\varphi$ and sells at $1$, getting a gain of 
$\varphi$. ONL's gain is 1. Both OPT and ONL used one trade, and 
OPT is a factor of $\varphi$ ahead of ONL.

Each of rounds 2 to $k-1$ are the same as the intermediate rounds
in Theorem~\ref{thm:ktradeLB}, with OPT gaining a factor of $\varphi^{2/3}$ 
ahead of ONL in each round.

Finally, in round $k$ we use the same price sequence in 
Theorem~\ref{thm:1tradeLB},
which gives an extra factor of $\varphi^{1-\eps}$. Note that ONL may have 
more trades left then OPT (in addition to the same reason as in 
Theorem~\ref{thm:ktradeLB}, in round 1 ONL may have done no trade), 
but again it is not useful for ONL.
\qed \end{proof}

\section{Bounded daily return, known duration}

Recall that in this model, the prices are bounded by 
$p(i)/\beta \le p(i+1) \le \alpha p(i)$ for some $\alpha, \beta>1$.
Trades can take place at days $0, 1, \ldots, T$.

\begin{theorem} \label{thm:dailyLB}
No deterministic algorithm has a competitive ratio better than
$\alpha^{T(2k \log \beta) / ((k+1) \log \beta + k \log \alpha)}$.
\end{theorem}
\begin{proof}
The adversary strategy is very simple and natural: whenever ONL is
not holding stock, the price goes up by a factor of $\alpha$ every day,
and while it is holding stock it goes down by $\beta$ every day.
Let the ONL trades be $(b_i, s_i)$, $i=1,\ldots,k$.
(If there are fewer than $k$ trades, dummy ones with $b_i=s_i$ can be
added.)
For $1 \le i \le k+1$, define $t_{2i-1} = b_i - s_{i-1}$ and
$t_{2i} = s_i - b_i$. (For convenience define $s_0=0$ and $b_{k+1}=T$.)
ONL's return is $1 / (\beta^{t_2} \beta^{t_4} \cdots \beta^{t_{2k}})$.
OPT's optimal actions, if allowed $k+1$ trades, is to hold during the
exact opposite intervals as ONL, i.e., buy at $s_i$ and sell at $b_{i+1}$
for $0 \le i \le k$. But since it can make at most $k$ trades,
its possible course of actions include skipping one of those trades, or making
one of the trades `span across two intervals', e.g., buying at $s_i$ and selling
at $b_{i+2}$. Thus OPT's return is one of
\[ \alpha^{t_3} \alpha^{t_5} \cdots \alpha^{t_{2k+1}},
\alpha^{t_1} \alpha^{t_5} \cdots \alpha^{t_{2k+1}},
\ldots,
\alpha^{t_1} \alpha^{t_3} \cdots \alpha^{t_{2k-1}}, \]
\[ \alpha^{t_1} \alpha^{t_3} \cdots \alpha^{t_{2k+1}}/\beta^{t_2},
\alpha^{t_1} \alpha^{t_3} \cdots \alpha^{t_{2k+1}}/\beta^{t_4},
\ldots,
\alpha^{t_1} \alpha^{t_3} \cdots \alpha^{t_{2k+1}}/\beta^{t_{2k}}.
\]

To attain the worst competitive ratio, these returns should be equal, which means
$t_1=t_3=\cdots=t_{2k+1}$ and $t_2=t_4=\cdots=t_{2k}$. This further implies
$\alpha^{k t_1} = \alpha^{(k+1)t_1}/\beta^{t_2}$,
which gives $\alpha^{t_1} = \beta^{t_2}$.
Together with $t_1+\cdots+t_{2k+1} = (k+1)t_1 + k t_2 = T$,
this gives
\[ 
t_1 = \frac{\log \beta}{(k+1)\log \beta + k \log \alpha}T, \quad
t_2 = \frac{\log \alpha}{(k+1)\log \beta + k \log \alpha}T
\] 
and thus the competitive ratio is at least
\[ \alpha^{k t_1} / (1/\beta^{k t_2}) = \alpha^{2k t_1} = 
\alpha^{T(2k \log \beta) / ((k+1) \log \beta + k \log \alpha)}. \]
\qed \end{proof}

\begin{algorithm}[h]
\caption{Static algorithm for known $\alpha, \beta$ and $T$.}
\label{alg:3}
\begin{algorithmic}
\STATE Set $t_1 = \frac{\log \beta}{(k+1)\log \beta + k \log \alpha}T$ and
$t_2 = \frac{\log \alpha}{(k+1)\log \beta + k \log \alpha}T$,
rounding to nearest integers.
\STATE Upon release of the $i$-th price $p(i)$:
\IF{$i=T$} 
  \IF{currently holding stock}
    \STATE sell at price $p(T)$ and the game ends
  \ENDIF
\ELSE
  \IF{have not been holding stock for $t_1$ days and not used up all trades} 
    \STATE buy at price $p(i)$
  \ELSIF{have been holding stock for $t_2$ days}
    \STATE sell at price $p(i)$
  \ENDIF
\ENDIF
\end{algorithmic}
\end{algorithm}

\begin{theorem} \label{thm:dailyUB}
Algorithm~\ref{alg:3} has competitive ratio 
$\alpha^{T(2k \log \beta) / ((k+1) \log \beta + k \log \alpha)}$.
\end{theorem}

\begin{proof}
In what follows we ignore the roundings on $t_1$ and $t_2$.
We argue that the worst case price sequence is exactly the one 
described in the proof of Theorem~\ref{thm:dailyLB}.
Consider an arbitrary day $i>0$.
Suppose at day $i$ ONL is holding stock or selling.
If $p(i-1)/\beta < p(i)$, i.e., the price change from day $i-1$ to $i$ is
not the maximum possible drop, we raise $p(i-1)$ to $\beta p(i)$ and make
a corresponding change to all earlier prices, i.e., multiply each of them
by a factor of $\beta p(i) / p(i-1)$. 

ONL's buy/sell decisions are unchanged as they do not depend on the 
prices at all. For any ONL trade completed (bought and sold)
on or before day $i-1$, their gains are
unaffected since both buying and selling prices are multiplied by the
same factor. For the one trade where it is holding or selling at day $i$,
ONL loses by a factor of $\beta p(i) / p(i-1)$
since the buying price of this trade is raised but the selling
price is not. All future trades are unaffected.
For the moment assume OPT's trading decisions also remain unchanged.
If OPT is holding or selling at day $i$, then it suffers
the same change as ONL, so the competitive ratio is unchanged.
Otherwise its gain is not affected and hence the competitive ratio
increases.

Similarly, suppose at day $i$ ONL is buying or is not holding stock.
If $p(i-1)\alpha > p(i)$, i.e., the price change from day $i-1$ to $i$ is
not the maximum possible rise, we lower $p(i-1)$ to $p(i)/\alpha$ and make
a corresponding change to all earlier prices, i.e., multiply each of them
by a factor of $p(i) / (\alpha p(i-1))$. 
For any ONL trade already completed before day $i$, its gain is
unaffected since both buying and selling prices are multiplied by the
same factor. All future trades are unaffected.
For OPT, if it is buying or holding at day $i-1$, then it gains 
from the lower buying price while the selling price (on or after day $i$)
is unchanged. Otherwise its gain is unaffected as in ONL.
Hence the competitive ratio can only increase.

Applying this to each day successively,
we can without loss of generality assume the prices follow a zig-zag pattern
as described in the proof of Theorem~\ref{thm:dailyLB}. 
Now, the optimal trades for this new price sequence may not be the same as
the original sequence, but that can only improve OPT's return.
The argument in the proof of Theorem~\ref{thm:dailyLB} establishes the
ratio between OPT and ONL for such a zig-zag price sequence. 
This proves the upper bound for our algorithm.
\qed \end{proof}

Algorithm 3 may feel unnatural since it does not depend 
on the price sequence at all
(this is called `static' in \cite{CKLW01}). 
But we prove that the following variation
of the algorithm has the same competitive ratio: it sells only when 
the current price falls below $h/\beta^{t_2}$ 
where $h$ is the highest price seen since the last purchase. 
This coincides with the `stop loss' strategy very common in real trading
(more precisely `trailing stop' \cite{GI95}
where the stop loss limit is not fixed but
tracks the highest price seen thus far, to potentially capture the most profit).

\begin{algorithm}[h]
\caption{Stop loss based algorithm for known $\alpha, \beta$ and $T$.}
\label{alg:4}
\begin{algorithmic}
\STATE Set $t_1$ and $t_2$ as in Algorithm~\ref{alg:3}.
\STATE Upon release of the $i$-th price $p(i)$:
\IF{$i=T$} 
  \IF{currently holding stock}
    \STATE sell at price $p(T)$ and the game ends
  \ENDIF
\ELSE
  \IF{have not been holding stock for $t_1$ days and not used up all trades} 
    \STATE buy at price $p(i)$
    \STATE set $h = p(i)$
  \ELSIF{currently holding stock}
    \STATE set $h = \max(h, p(i))$
    \STATE sell at price $p(i)$ if $p(i) < h / \beta^{t_2}$
  \ENDIF
\ENDIF
\end{algorithmic}
\end{algorithm}
 
\begin{theorem}
Algorithm~\ref{alg:4} has competitive ratio 
$\alpha^{T(2k \log \beta) / ((k+1) \log \beta + k \log \alpha)}$.
\end{theorem}
\begin{proof}
Recall that $\alpha^{t_1} = \beta^{t_2}$. Let $r$ denote this common 
value, and 
the competitive ratio we want to prove is then equal to $r^{2k}$.
Roughly speaking, our approach is to partition the time horizon so that, in each 
partition, $x$ trades in OPT are associated to $y$ trades in ONL such that 
the ratio between their gains is at most $r^{x+y}$. Since each of OPT and
ONL makes at most $k$ trades, this proves the theorem.

Suppose ONL completed $\ell$ trades, $\ell \le k$.
Denote by $(b_i, s_i)$ ONL's $i$-th trade.
For technical reasons, also define $b_0=s_0=0$ and $b_{\ell+1}=s_{\ell+1}=T$.
Let $H_i = [b_i, s_i]$, $0 \le i \le \ell$, be the $i$-th holding period,
i.e., the (closed) time interval where ONL is holding stock 
in its $i$-th trade,
and let $N_i = (s_{i-1}, b_i)$, $1 \le i \le \ell+1$,
be the (open) $i$-th non-holding period.
$H_0$ is the holding period for the trivial trade $(b_0, s_0)$ 
defined for analysis purposes only.
If $H_{\ell}$ does not contain time $T$ (i.e., ONL is not holding stock till
the end where a forced sale happens), also define 
$H_{\ell+1} = [b_{\ell+1}, s_{\ell+1}]$ for the trivial trade
$(b_{\ell+1}, s_{\ell+1})$.
Each $N_i$ has length exactly $t_1$, and each $H_i$ (except $H_0$ and
$H_{\ell+1}$) has length at least $t_2$.
Let $h_i$ be the highest price during $H_i$.
We first show the following properties:

\begin{enumerate}

\item[(1)] For any two days $x$ and $y$ in the same $H_i$, where $x<y$, 
we have $p(y) \ge p(x)/r$. 
This is because at day $y$ the highest price seen thus far is at least $p(x)$, 
and so the stop loss threshold is at least $p(x)/r$. 
(In fact $p(y)$ could be up to one factor of $\beta$
smaller since ONL sells as soon as this limit is reached, but this difference
can be ignored.) As a direct consequence, $p(s_i) \ge h_i/r$.
(If day $T$ arrives before the threshold is reached then
the final sale is made at a price higher than this.)

\item[(2)] Without loss of generality we can assume 
$p(b_{i+1}) = p(s_i) r$.
The reason is that during any $N_i$, the price sequence
can be transformed so that it goes up by $\alpha$ every day, 
by the same argument as in Theorem~\ref{thm:dailyUB}, since
ONL always waits the same number of days before the next purchase and is
independent of price changes during this period.
Then the price rises by a factor of $\alpha^{t_1}$ over this non-holding period.

\item[(3)] OPT would not buy or sell strictly within any $N_i$. 
This follows from the argument in (2), since it is easy to see that OPT 
only buys at local minima and sells at local maxima.

\end{enumerate}

Consider an OPT trade $(b^*,s^*)$.
Suppose $b^*$ falls within $H_u$ and $s^*$ falls within $H_v$, where $v \ge u$. 
Its gain $g^*$ is equal to
\[
g^* = \frac{p(s^*)}{p(b^*)} \le \frac{h_v}{p(b_u)/r}
\le \frac{p(s_v) r }{ p(b_u)/r }
= \frac{p(s_v)}{p(b_u)} r^2
\]
where the inequalities are due to (1).
Note that if $u=0$, then $p(b^*)=p(b_u)$ and thus a factor of $r$ can be
removed from the above bound; similarly if $v=\ell+1$ then $p(s_v)=h_v$
and another factor of $r$ can be removed.
Thus, if we define an indicator variable $I_i$ which is 0 if $i=0$ or $i=\ell+1$
and 1 otherwise, then
\[ g^* \le \frac{p(s_v)}{p(b_u)} r^{I_u+I_v} \]

Then
\[
g^* \le \frac{p(s_v)}{p(b_v)} \frac{p(b_v)}{p(s_{v-1})} 
  \frac{p(s_{v-1})}{p(b_{v-1})} \frac{p(b_{v-1})}{p(s_{v-2})} \ldots
  \frac{p(s_{u+1})}{p(b_{u+1})} \frac{p(b_{u+1})}{p(s_u)} \frac{p(s_u)}{p(b_u)} 
  r^{I_u+I_v}
\] \[ 
= \frac{p(s_v)}{p(b_v)} r \frac{p(s_{v-1})}{p(b_{v-1})} r \ldots 
r \frac{p(s_u)}{p(b_u)} r^{I_u+I_v}
= r^{v-u+I_u+I_v} \prod_{i=u}^v g_i \]
where $g_i = p(s_i)/p(b_i)$ is the gain of the $i$-th ONL trade.
The second last equality is due to (2).

If no other OPT trade sells during $H_u$ or buys during $H_v$,
we can associate this OPT trade with these $v-u+1$ ONL trades
$(b_u, s_u), \ldots, (b_v, s_v)$. 
No other OPT trades would be associated with these ONL trades.
But if other OPT trades fall within $H_u$ or $H_v$ then this cannot be done
directly. Instead we consider groups of OPT trades, determined as follows.
Suppose there are two OPT trades $(b^*_1,s^*_1)$ and 
$(b^*_2, s^*_2)$ such that 
the earlier one sells at the same holding period as the later one buys, i.e., 
$b^*_1$ is in $H_u$, 
$s^*_1$ and $b^*_2$ are in the same $H_v$,
and $s^*_2$ is in $H_w$, where $u \le v \le w$. 
Then $p(b^*_2) \ge p(s^*_1)/r$ due to (1), and so
$(p(s^*_2)/p(b^*_2))(p(s^*_1)/p(b^*_1)) \le (p(s^*_2)/p(b^*_1)) r$,
and by losing a factor of $r$ we can replace the two OPT trades 
with one $(b^*_1,s^*_2)$. 

By repeatedly applying the argument, we can partition the time horizon into
disjoint parts: in part $j$, OPT has $x_j$ trades $(b^*_{j,1}, s^*_{j,1}), \ldots,
(b^*_{j,x_j}, s^*_{j,x_j})$, $b^*_{j,1}$ falls within some $H_{u_j}$ and 
$s^*_{j,x_j}$ falls within some $H_{v_j}$, 
and no other OPT trade sells during $H_{u_j}$ or buys during $H_{v_j}$ (otherwise
they would have been merged into this as well).
ONL made a total of $y_j=v_j-u_j+1$ trades here, each of gain $g_{j,i}$.
The $x_j$ OPT trades, each of gain $g^*_{j,i}$, have been merged into one trade 
$(b^*_{j,1}, s^*_{j,x_j})$ of gain $g^*_j$, such that
$g^*_j \ge \prod_i g^*_{j,i}/r^{x_j-1}$. 
Since $g^*_j \le r^{v_j-u_j+I_{u_j}+I_{v_j}} \prod_{i=u_j}^{v_j} g_{j,i}$, we have 
\[
\prod_i g^*_{j,i} \le r^{x_j-1} r^{v_j-u_j+I_{u_j}+I_{v_j}}  
   \prod_{i=u_j}^{v_j} g_{j,i}
= r^{x_j+y_j} r^{I_{u_j}+I_{v_j}-2} \prod_{i=u_j}^{v_j} g_{j,i}
\]

Over all partitions, therefore, the overall gains ratio between OPT and ONL is 
$\prod_j \left( \prod_i g^*_{j,i} / \prod_{i=u_j}^{v_j} g_{j,i} \right)
\le \prod_j r^{x_j+y_j} r^{I_{u_j}+I_{v_j}-2}$.

If none of $u_j$ or $v_j$ involve the two trivial holding periods $H_0$ and
$H_{\ell+1}$, then all $I_{u_j}$ and $I_{v_j}$ are equal to 1, and
$\sum_j x_j \le k$, $\sum_j y_j \le \ell$. Hence
\[ \prod_j r^{x_j+y_j} r^{I_{u_j}+I_{v_j}-2}
\le r^{k+\ell} r^{1+1-2} \ldots r^{1+1-2} \le r^{2k} \]

If one trivial holding period is involved, e.g., $b^*_{1,1}$ falls into $H_0$
and thus $u_1=0$, then $I_{u_1}=0$ and $\sum_j y_j \le \ell+1$. Then
\[ \prod_j r^{x_j+y_j} r^{I_{u_j}+I_{v_j}-2}
\le r^{k+\ell+1} r^{0+1-2} r^{1+1-2} \ldots r^{1+1-2} \le r^{2k} \]

Finally if both trivial holding periods are involved, i.e., $b^*_{1,1}$ falls into 
$H_0$ and the last sale falls into $H_{\ell+1}$, then
\[ \prod_j r^{x_j+y_j} r^{I_{u_j}+I_{v_j}-2}
\le r^{k+\ell+2} r^{0+1-2} r^{1+1-2} \ldots r^{1+1-2} r^{1+0-2} \le r^{2k} \]
\qed \end{proof}

\section{Conclusion}

There are still many possible directions that can be explored. Some examples
include randomized algorithms, allowing short-selling, trading two or more
stocks, or allowing the money to be split into a small number of bets 
(e.g., the investor can sell half of its shares first, and potentially let
the other half to earn more profit). The last one models the intermediate case 
between the splittable and unsplittable versions of the problem.

\end{document}